\documentclass[12pt,reqno,a4paper]{amsart}
\usepackage{amsmath,amssymb,amsfonts}
\usepackage[mathscr]{eucal}
\usepackage{amsthm}
\usepackage[arrow]{xy}
\usepackage{stmaryrd}
\usepackage[margin=1.2in]
{geometry}

\newtheorem{proposition}{Proposition}
\newtheorem{theorem}{Theorem}

\newtheorem{lemma}{Lemma}

\theoremstyle{definition}
\newtheorem{remark}{Remark}
\newtheorem{example}{Example}[section]
\newtheorem{definition}{Definition}

\newcommand{\ci}{C^\infty}

\newcommand{\C}{\mathcal}
\newcommand{\p}{\partial}
\renewcommand{\L}{L_\infty}
\newcommand{\F}[2]{\frac{#1}{#2}}

\renewcommand{\ker}{\mathrm{ker}}

\newcommand{\Vect}{\mathrm{Vect}}

\title{Homotopy Loday Algebras and Symplectic $2$-Manifolds}

\author{Matthew.~T.~Peddie}

\address{Department of Mathematics,  The Pennsylvania State University,
State College, 16801, U.S.}
\email{mtp27@psu.edu, matt.peddie11@gmail.com}

\keywords{Nonabelian derived bracket, Loday algebra, Leibniz algebra, homotopy algebra, Dorfman bracket, Courant bracket}
\subjclass[2010]{53D17, 53Z05, 58A50}

\begin{document}

\begin{abstract}
Using the technique of higher derived brackets developed by Voronov, we construct a homotopy Loday algebra in the sense of Ammar and Poncin  associated to any symplectic $2$-manifold. The algebra we obtain has a particularly nice structure, in that it accommodates the Dorfman bracket of a Courant algebroid as the binary operation in the hierarchy of operations, and the defect in the symmetry of each operation is measurable in a certain precise sense. We move to call such an algebra a homotopy Dorfman algebra, or a $D_\infty$-algebra, which leads to the construction of a homotopy Courant algebroid.
\end{abstract}

\maketitle

\section{Introduction}
The concept of a derived bracket was introduced by Koszul (unpublished) and formalised by Kosmann-Schwarzbach \cite{Kos96}, where given a graded Lie algebra $L$ and an odd derivation $D$ of this Lie algebra, a secondary bracket
	\begin{equation*}
	[a,b]_D := (-1)^a[Da,b],\quad a,b,\in L,
	\end{equation*}
may be defined from the original Lie bracket $[-,-]$ on $L$. This derived bracket satisfies a Jacobi identity if $D$ satisfies the non-trivial nilpotency condition  $D^2=0$, however it does not possess any symmetry properties in general, and so does not define a Lie bracket. The Lie algebra $L$ equipped with this bracket is an example of a Loday algebra, or a Leibniz algebra as they were originally named \cite{Lod93}. Since their introduction, derived brackets have appeared (or rather have been identified) in many areas of geometry and physics: the Schouten-Nijenhuis bracket of multivector fields \cite{Sch40}, the Koszul bracket on differential forms on a Poisson manifold \cite{Fuc82, Kos85}, odd and even Poisson brackets on (super)manifolds \cite{Vor02}, and the Dorfman bracket on sections of a Courant algebroid \cite{Roy99}, to mention only a few. An excellent survey of derived brackets can be found in the article \cite{Kos04} together with many more references.

Around the same time as the introduction of derived brackets, (strongly) homotopy Lie algebras, or $\L$-algebras, began appearing in physical literature in relation to gauge field theories. In the works \cite{Vor05a,Vor05b}, Voronov introduced a higher derived bracket construction for producing such $\L$-algebras. Starting with a Lie (super)algebra $L$ and an odd derivation $D$, higher order operations $[\cdots]_D$ on $n$ arguments were introduced on an abelian subalgebra $V$ of $L$ with the assistance of a projector $P$, via the formula
	\begin{equation*}
   [a_1,\ldots,a_n]_D:= P[\cdots[Da_1,a_2],\ldots,a_n],
    \end{equation*}
for $a_1,\ldots,a_n\in L$. The abelian condition imposed on $V$ ensures the symmetry of these operations, which were shown to endow $V$ with the structure of an $\L$-algebra if $D^2=0$.

Since the introduction of this higher derived bracket construction, extensions have appeared offering techniques for producing $\L$-operations when $V$ is no longer abelian \cite{Ban15}, and rephrasing the entire construction using homotopy data and the language of coalgebras. Interestingly, it wasn't until later that the notion of a homotopy Loday algebra was introduced  \cite{AP10}, whence  Uchnio \cite{Uch11} gave a construction of such homotopy Loday algebras (named homotopy Leibniz algebras) via derived brackets. This construction simply removes the projector and hence the abelian condition from Voronov's construction which removes the symmetry of the defined operations.

In this article we construct a specific class of homotopy Loday algebras which we call homotopy Dorfman algebras, or $D_\infty$-algebras. The reason for this is that the Dorfman bracket of a Courant algebroid sits naturally in this hierarchy of operations, which extend the Dorfman bracket in the same way that higher $\L$-brackets extend the binary Lie bracket. Though each one of these higher Dorfman brackets  is a Loday bracket, and so possesses no symmetry, the discrepancy in the symmetry is measurable by the introduction of a bilinear form, and so these operations can be seen to act on the space of sections of a pseudo-Euclidean vector bundle. The skew-symmetrisation of these operations produce higher Courant brackets, and the operations introduced here lead naturally to the construction of homotopy Courant algebroids which will be the subject of a forthcoming article. We mention that higher (homotopy) Courant brackets appear already in the literature \cite{Ber07} (see also \cite{KW15} where they are remarked upon) though defined differently, and without reference to an algebroid. We view our definition of the Courant brackets as natural, being skew-symmetrisations of the naturally defined Dorfman operations. Other concepts of higher Courant brackets are known, for instance in \cite{Zam12}, but whose setting does not admit a homotopy analogue.

Throughout this we work in the super category where it will be convenient to omit the prefix super; when we write vector space we refer to a super vector space and so on. As is usual, the $\mathbb{Z}_2$-grading will be called the parity, and denoted by $\tilde{a}$ when we wish to be explicit about the parity of $a$. In formula however, it will be enough to write $(-1)^a$ for instance, when we mean $(-1)$ raised to the parity of $a$.

\section{Set-Up and Main Statement}
Let $E\rightarrow M$ be a pseudo-Euclidean vector bundle over a (super)manifold $M$, and denote by $g:\Gamma(E)\times\Gamma(E)\rightarrow \mathbb{R}$ the symmetric non-degenerate bilinear form. As was shown in \cite{Roy02}, pseudo-Euclidean vector bundles are in one to one correspondence with non-negatively graded manifolds $\C{M}$ equipped with a weight $2$ symplectic form, or symplectic $2$-manifolds. The weight of the symplectic form forces the fibred structure
    \begin{equation*}
    \C{M}\rightarrow \Pi E\rightarrow M,
    \end{equation*}
of the symplectic $2$-manifold, where $\Pi E$ is the vector bundle $E$ with shifted parity in the fibres: $(\Gamma(\Pi E))_0 := \Gamma(E)_1$ and $(\Gamma(\Pi E))_1 := \Gamma(E)_0$ where $\Gamma(E) = \Gamma(E)_0\oplus \Gamma(E)_1$ under the $\mathbb{Z}_2$-grading. The algebra of functions on the graded manifold $\C{M}$, denoted by $\C{A}$, is itself naturally non-negatively graded, and decomposes $\C{A}=\oplus_{k\geq 0}\C{A}^k$ into the homogeneous weighted components. A natural example of such a manifold is the cotangent bundle $T^*\Pi E$ to the total space of $\Pi E$.

The weight $2$ symplectic form on $\C{M}$ gives rise to a non-degenerate weight $-2$ Poisson bracket which we will denote by $[-,-]$. If one fixes a linear connection $\nabla$ in $E$, the symplectic structure on $\C{M}$ coincides with that described in \cite{Rot91}, in which case the Poisson bracket carries the expressions
	\begin{gather*}
    [X,f]_\nabla = X(f),\qquad [X,u]_\nabla = \nabla_Xu,\\
    [u,v]_\nabla = (-1)^ug(u,v),\quad [X,Y]_\nabla = [X,Y]_c + R_\nabla (X,Y),\nonumber
    \end{gather*}
where $u,v\in\Gamma(E)$, $X,Y\in \Vect(M)$ are weight $2$ vector fields and $[X,Y]_c$ is the usual commutator of vector fields, and $R_\nabla$ is the curvature $2$-form of $\nabla$. A linear connection is not needed in general, and we will not refer to this again, however it is beneficial to view the above expressions for the bracket. The space of sections $\Gamma(E)$ may be naturally identified with the space of linear functions $\C{A}^1$ on $\Pi E$ via the odd isomorphism
    \begin{equation}\label{Section 1 Isomorphism}
	\chi:\Gamma(E)\rightarrow \C{A}^1,\quad u\mapsto \chi(u) = \chi_u,\qquad u\in\Gamma(E),
	\end{equation}
where $\tilde{\chi_u} = \tilde{u} + 1$, and $\chi$ identifies a section $u$ with a function on $\Pi E^*$, and then raises the indices via the metric $g$. This map relates the non-degenerate form $g$ on $\Gamma(E)$ to the Poisson bracket restricted to $\C{A}^1$ by
    \begin{equation}\label{Section 1 chi relating the form g to the bracket of functions}
    g(u,v) = (-1)^u[\chi_u,\chi_v],
    \end{equation}
from which it is clear that the space $\C{A}^0\oplus \C{A}^1$ forms a non-trivial Poisson subalgebra of $\C{A}$.
Define a projector $P$ on the algebra $\C{A}$ by
	\begin{equation}\label{Section 1 Projector P}
	P:\C{A}\rightarrow \C{A}^0\oplus \C{A}^1,\qquad P^2 = P,
	\end{equation}
singling out the terms in a function $f\in\C{A}$ of weight zero and weight one. The projector in eq. \eqref{Section 1 Projector P} satisfies the ``distributivity law''
	\begin{equation}\label{Section 1 Distr Law for projection}
	P[f,g] + [Pf,Pg] = P[Pf,g] + P[f,Pg],
	\end{equation}
for functions $f,g\in\C{A}$, which is not hard to obtain by expanding the functions in terms of their weight.

Now let $Q:\C{A}\rightarrow \C{A}$ be an odd derivation of the Poisson algebra $\C{A}$, i.e.
	\begin{equation}\label{Section D is a derivation of the bracket}
	Q[f,g] = [Qf,g] + (-1)^{f}[f,Qg].
	\end{equation}
Paralleling the definition given in \cite{Vor05a,Vor05b}, we define higher operations on $\Gamma(E)$ generated by the derivation $Q$ via the isomorphism in eq. \eqref{Section 1 Isomorphism}.

\begin{definition}
The $k$th higher derived bracket $d_k:\Gamma(E)^{\times k}\rightarrow \Gamma(E)$ is the multilinear operation defined by the higher derived bracket formula
	\begin{align}\label{Section 1 Definition Derived bracket formula}
	\chi(d_k(u_1,\ldots,u_k)) := &(-1)^{(k-1)u_1 + (k-2)u_2+\cdots+ u_{k-1}}P[\cdots[Q\chi_{u_1},\chi_{u_2}],\ldots,\chi_{u_k}] \\
	& - (-1)^{(k-1)u_1 + (k-2)u_2+\cdots +u_{k-1}}[P[\cdots[Q\chi_{u_1},\chi_{u_2}],\ldots,\chi_{u_{k-1}}],\chi_{u_k}],\nonumber
	\end{align}
where $u_1,\ldots,u_k\in\Gamma(E)$.
\end{definition}

These operations carry parity $k\!\mod 2$, and that we obtain an element of $\Gamma(E)$ upon the projection of these terms follows by expanding $Q$ in terms of the weight and comparing the degrees. One can see that the second term is strictly of weight $0$, and simply removes the weight $0$ terms from the first.
\begin{remark}
 Everything that follows may be defined for an arbitrary odd derivation of the Poisson algebra $\C{A}$. It is convenient however to suppose that $Q$ arises as an inner derivation $Q = [\theta,-]$ for $\theta\in\C{A}$, which is automatically compatible with the Poisson structure on $\C{A}$. We assume this for the remainder of the text, and fix $\theta\in\C{A}$ as an arbitrary odd inhomogeneous function defining $Q$. It will also be convenient to assume  that $P\theta = 0$, if only to simplify some of the later calculations.
\end{remark}
 For an arbitrary derivation we need to impose a secondary condition on $Q$, requiring that $Q$ preserves the subalgebra $\ker P$ (notice that this is indeed a subalgebra as a consequence of the weight). In which case we see that $Q$ satisfies the condition
	\begin{equation}\label{Section 1 Q preserves the subalgebra ker P}
	PQP = PQ.
	\end{equation}
However with our assumption that $Q = [\theta,-]$ is an inner derivation, together with the condition that $P\theta = 0$, this preservation condition is satisfied automatically, since $P\theta = 0$ guarantees that the weight of $\theta$ is greater that $1$, and hence that the weight of $Q$ is non-negative. Therefore if $f\in\ker P$, the weight of $Q$ forces $Qf\in\ker P$ also.

These higher operations are neither symmetric nor skew-symmetric, providing the difference between these and other higher operations constructed via higher derived brackets in the literature. The symmetric component of these operations can however be ``measured'' by the bilinear form $g$ on $\Gamma(E)$.

\begin{proposition}\label{Section 1 Proposition symmetry}
The higher operations $d_k:\Gamma(E)^{\times k}\rightarrow \Gamma(E)$ satisfy the symmetry condition
    \begin{align*}
    & \chi\left(d_k(u_1,\ldots,u_{i},u_{i+1},\ldots,u_k)+(-1)^{u_{i}u_{i+1}}d_k(u_1,\ldots,u_{i+1},u_{i},\ldots,u_k)\right)\\
    & \quad= (-1)^{\varepsilon+(u_i+u_{i+1})(u_{i+2}+\cdots+u_k + (k-i-1))}P[\cdots[Q\chi_{u_1},\chi_{u_2}],\ldots,\chi_{u_k}],[\chi_{u_i},\chi_{u_{i+1}}]]\\
    & \quad\quad - (-1)^{\varepsilon+(u_i+u_{i+1})(u_{i+2}+\cdots+u_k + (k-i-1))}W[\cdots[Q\chi_{u_1},\chi_{u_2}],\ldots,\chi_{u_k}],[\chi_{u_i},\chi_{u_{i+1}}]],
    \end{align*}
where $\varepsilon = (k-1)u_1 + (k-2)u_2+\cdots+ u_{k-1}$, and $W:\C{A}\rightarrow \C{A}^0$ is a secondary projector to the trivial Poisson subalgebra $\C{A}^0$.
\end{proposition}
Notice that, due to degree reasons,  $[Pf,g] = W[f,g]$ for any function $f\in\C{A}$ and linear function $g\in\C{A}^1$, which allows us to rewrite the second term in eq. \eqref{Section 1 Definition Derived bracket formula} with the projector $P$ inside in terms of $W$. The proof of proposition \ref{Section 1 Proposition symmetry} follows from applications of the Jacobi identity whilst keeping careful track of the signs; we omit it.

\begin{definition}
Define the $k$th defect map $D_k:\Gamma(E)^{\times k-2}\times\ci(M)\rightarrow \Gamma(E)$ by the formula
	\begin{align*}
	\chi\left(D_k(u_1,\ldots,u_{k-2}|f)\right) & := (-1)^{u_1(k-3) + \cdots + u_{k-3}}P[\cdots[Q\chi_{u_1},\chi_{u_2}],\ldots,\chi_{u_{k-2}}],f]\\
    & \quad\quad - (-1)^{u_1(k-3) + \cdots + u_{k-3}}W[\cdots[Q\chi_{u_1},\chi_{u_2}],\ldots,\chi_{u_k}],f],
	\end{align*}
for sections $u_1,\ldots,u_{k-2}\in\Gamma(E)$ and a function $f\in\ci(M)$.
\end{definition}

The defect maps may be used (by definition) to measure the discrepancy in the skew-symmetry of the higher brackets. The additional signs in proposition \ref{Section 1 Proposition symmetry} combine with those in $\epsilon$, leaving only that necessary to define the form $g$ in terms of $\chi$:
	\begin{align*}
	&d_k(u_1,\ldots,u_{i},u_{i+1},\ldots,u_k) +(-1)^{u_{i}u_{i+1}}d_k(u_1,\ldots,u_{i+1},u_{i},\ldots,u_k)\\
	&\quad\quad = (-1)^{(u_i+u_{i+1})(u_{i+2} + \cdots + u_k)}D_k(u_1,\ldots,\hat{u}_i,\hat{u}_{i+1},\ldots,u_k | g(u_i,u_{i+1})),
	\end{align*}
where the hat terms denote omission from the sequence of sections.
In fact the map for $k=2$ is well-known and was introduced in the work \cite{LWX97}.
\begin{example}
Suppose that $\theta\in\C{A}^3$ is an odd weight $3$ function on $\C{M}$. Then $Q = [\theta,-]$ is of weight $1$, and the only non-trivial derived bracket is $d_2$, where
	\begin{equation*}
	\chi\left(d_2(u,v)\right) = (-1)^u[Q\chi_u,\chi_v],
	\end{equation*}
for $u,v\in\Gamma(E)$. This bracket coincides with the standard Dorfman bracket on the sections of a pseudo-Euclidean vector bundle (of which the Courant bracket is a skew-symmetrisation of). The map $D_2:\ci(M)\rightarrow\Gamma(E)$ measuring the defect in the symmetry of the Dorfman bracket satisfies the relation
	\begin{equation*}
	d_2(u,v) + (-1)^{uv}d_2(v,u) = D_2g(u,v),
	\end{equation*}
defined in \cite{LWX97} via the de Rham differential on $M$ and the anchor map of a Courant algebroid.
\end{example}
The example shows that the Dorfman bracket sits as the binary bracket in this hierarchy, and following that for $\L$-algebras, we will call the operations defined in eq. \eqref{Section 1 Definition Derived bracket formula} the higher derived Dorfman brackets on the space of sections $\Gamma(E)$. We now show that these Dorfman brackets satisfy the higher Jacobi identities of a homotopy Loday algebra.

A homotopy Loday algebra, or a homotopy Leibniz algebra, is a relaxed $\L$-algebra, requiring a sequence of higher order brackets which satisfy higher Jacobi identities, but which no longer satisfy any (skew-)symmetry conditions. The category of homotopy Loday algebras is investigated in the work \cite{AP10}, and such algebras were explicitly constructed on an arbitrary Lie algebra $L$ in \cite{Uch11} via derived brackets without the use of projectors.

\begin{definition}\label{Section 1 Definition sh Leibniz algebra}
A vector space $V$ together with a collection of multilinear operations $l_i:V^{\times i}\rightarrow V$ carrying parity $k\!\mod 2$ determines a homotopy Loday algebra when
	\begin{gather}\label{Section 1 nth Loday identity}
	J_n(v_1,\ldots,v_n) = \sum_{i+j=n+1}\sum_{k\geq j}\sum_{\sigma\in Sh(k-j,j-1)} K(\sigma)(-1)^{(k+1)(j+1)}(-1)^{j(v_{\sigma(1)} + \cdots +v_{\sigma(k-j)})}\\
	\qquad l_i(v_{\sigma(1)},\ldots,v_{\sigma(k-j)},l_j(v_{\sigma(k-j+1)},\ldots,v_{\sigma(k-1)},v_k),v_{k+1},\ldots,v_{i+j-1})=0,\nonumber
	\end{gather}
where $Sh(k-j,j-1)$ is the set of $(k-j,j-1)$ unshuffles, and $K(\sigma) = sgn(\sigma)\kappa(\sigma)$ is the product of the parity of the permutation with the Koszul sign $\kappa(\sigma)$ obtained from permuting the elements.
\end{definition}
Our main statement then is as follows.

\begin{theorem}
The $n$th higher Loday identity eq. \eqref{Section 1 nth Loday identity} for the higher Dorfman brackets defined from the derivation $Q$ is equivalent to the $n$th derived Dorfman bracket of $Q^2$. In particular, if $Q$ satisfies $Q^2=0$, then the higher Dorfman brackets endow $\Gamma(E)$ with a homotopy Loday algebra structure.
\end{theorem}
As a consequence of the classical Dorfman bracket nesting as the binary bracket of this hierarchy of operations, we call the space $(\Gamma(E),g,d_1,d_2,d_3,\dots)$ a $D_\infty$-algebra, or a homotopy Dorfman algebra. Such algebras should be defined as homotopy Loday algebras whose symmetry defect in the operations can be measured by some bilinear form. It remains to show that such operations satisfy the higher Loday identities in eq. \eqref{Section 1 nth Loday identity}. It is instructive to view the first few independently; for $n=1$,
	\begin{equation*}
	\chi_{d^2_1(u)} = PQ(PQ\chi_u) = PQ^2\chi_u,
	\end{equation*}
using eq. \eqref{Section 1 Q preserves the subalgebra ker P}. The case for $n=2$ is more interesting and illustrates the manipulations involved in the proof of the general case. The second order identity is
	\begin{equation*}
	J_2(u,v) = d_2(d_1u,v) + (-1)^ud_2(u,d_1v) - d_1d_2(u,v),
	\end{equation*}
the derivation property for the first bracket $d_1$ treated as a differential. Then
	\begin{align*}
	\chi(J_2(u,v)) & = (-1)^{u+1}PQ(P[Q\chi_u,\chi_v] - [PQ\chi_u,\chi_v])\\
	& \quad +(-1)^{u+1}(P[QPQ\chi_u,\chi_v] - [PQPQ\chi_u,\chi_v])\\
    & \quad \quad + (P[Q\chi_u,PQ\chi_v] - [PQ\chi_u,PQ\chi_v]),
	\end{align*}
from which, if one uses the properties in eq. \eqref{Section 1 Distr Law for projection} and eq. \eqref{Section D is a derivation of the bracket} for the projector and the derivation, we obtain	
	\begin{equation*}
	= (-1)^{u+1}P[Q^2\chi_u,\chi_v] + (-1)^u[PQ^2\chi_u,\chi_v] = \chi(-d_2(u,v)_{Q^2}).
	\end{equation*}
We defer the proof of the general Loday identity until the end of the article.

In the same way that the Courant bracket is a skew-symmetrisation of the Dorfman bracket, we can define higher Courant brackets on $\Gamma(E)$.
\begin{definition}
Define the $k$th Courant bracket on $\Gamma(E)$ by the skew-symmetrisation of the $k$th Dorfman bracket $d_k$:
	\begin{equation*}
	[u_1,\ldots,u_k]_C = \F{1}{n!}\sum_{\sigma\in S_n}K(\sigma)d_k(u_{\sigma(1)},\ldots,u_{\sigma(n)}).
	\end{equation*}
\end{definition}
Paralleling the binary case, these higher identities do not satisfy the higher Jacobi identities, but rather hold up to some defect term. One can expect to obtain a sequence of $E$-valued forms $T_3\in\Gamma(\wedge^3 E)$, $T_4\in\Gamma(\wedge^4 E)$, and so on as for the usual Courant case. In the work \cite{RW98} the tensor $T_3$ was used to produce an $\L$-algebra associated to any Courant algebroid. It would be interesting to see whether these higher tensors could be incorporated to produce such a structure associated to a homotopy Courant algebra.

\section{Dorfman Brackets for $P_\infty$-Manifolds}
For a manifold $M$, the vector bundle $\C{T} = TM\oplus T^*M$ is pseudo-Euclidean with the non-degenerate form given by the canonical pairing of $1$-forms and vector fields on $M$. The associated symplectic $2$-manifold $\C{M}$ is the cotangent bundle $T^*\Pi TM$. Introduce local coordinates $x^a,p_a,\xi^a,\pi_a$ on $\C{M}$ with weights $0,2,1,1$ respectively, inherited from the double vector bundle structure of $T^*\Pi TM$. The weight acts as a total weight of a bigrading, also specified by the double vector bundle structure, with the weight given by the Euler vector field
	\begin{equation*}
	w = 2p_a\F{\p}{\p p_a} + \xi^a\F{\p}{\p\xi^a} + \pi_a\F{\p}{\p\pi_a},
	\end{equation*}
and the bigrading by the two vector fields
	\begin{equation*}
	\varepsilon_1 = p_a\F{\p}{\p p_a} + \xi^a\F{\p}{\p\xi^a},\qquad \varepsilon_2 = p_a\F{\p}{\p p_a} + \pi_a\F{\p}{\p\pi_a}.
	\end{equation*}
The algebra of functions $\C{A}$ decomposes with respect to both the weight and the bigrading.

It is well-known that this structure is sufficient to give rise to the Dorfman bracket generated by the de Rham differential $d$ viewed as a homological vector field on $\Pi TM$. This vector field has the associated Hamiltonian function $\Delta = \xi^ap_a$ of weight $3$ and bidegree $2,1$, and defines the Dorfman bracket on $\Gamma(\C{T})$, 
	\begin{equation*}
    d_2(X+\eta,Y+\tau) = [X,Y]_c + \C{L}_X\tau - (-1)^{\eta Y}\imath_Y\eta,
    \end{equation*}
for vector fields $X,Y$ and $1$-forms $\eta,\tau$ on $M$. Due to the weight, the function $\Delta$ cannot generate any higher structure. However if $M$ carries homotopy Poisson structure, or $P_\infty$-structure, then naturally defined higher Dorfman brackets arise.

A homotopy Poisson structure on a manifold is given by introducing an even multivector field $P\in\ci(\Pi T^*M)$ such that $[P,P]_{SN} = 0$ under the Schouten-Nijenhuis bracket. The Poisson brackets on $M$ are given by derived brackets
	\begin{equation*}
	\{f_1,\ldots,f_k\}:=\pm [\cdots[P,f_1]_{SN},\dots, f_k]_{SN}\big|_{M},
	\end{equation*}
whose higher Jacobi identities are equivalent to the Poisson-nilpotency of $P$. If $P$ is a bivector field then this recovers the usual binary Poisson bracket, which gives rise to the Koszul bracket of differential forms, an odd binary bracket in the algebra $\ci(\Pi TM)$. In the same sense, a homotopy Poisson structure generates a sequence of higher Koszul brackets giving the tangent bundle $TM$ the structure of an $\L$-bialgebroid \cite{KV08}.

To obtain the higher Koszul structure, notice that $P$ defines a homological vector field $d_P$ on $\Pi T^*M$ which in turn defines a linear Hamiltonian function on $T^*\Pi T^*M$. It remains to apply the canonical isomorphism of double vector bundles $T^*\Pi T^*M \cong T^*\Pi TM$, which identifies an odd Poisson-commuting function $K_P$ on $T^*\Pi TM$. This function $K_P$ defines the higher Koszul brackets on $\Pi TM$ via derived brackets, and in general is inhomogeneous in weight, though it satisfies $\C{L}_{\varepsilon_1}K_P = K_P$ by construction, and so carries a fixed $\varepsilon_1$ grading of $1$. One can consider the sum of the functions $K_P$ and $\Delta$, and use this to define a higher Dorfman structure, analogous to the Dorfman structure arising from a Lie bialgebroid.

Consider the space of sections $\Gamma(\wedge T^*M)$ endowed with the de Rham differential $d$. On $\Gamma(\wedge T^*M)$ we can define the usual insertion and Lie derivative operations $\imath_X$ and $\C{L}_X$ for a vector field $X$ on $M$, either directly via their action on tensors, or in terms of the Poisson algebra $\C{A}$ and the map in eq. \eqref{Section 1 Isomorphism}. Set
    \begin{gather*}
    \chi(\imath_X\eta):=(-1)^X[\chi_X,\chi_\eta],\qquad \chi(d\eta):= [\Delta,\chi_\eta],\\
    \chi(\C{L}_X\eta):=(-1)^X[[\Delta,\chi_X],\chi_\eta],
    \end{gather*}
for a vector field $X$, a form $\eta\in\Gamma(\wedge^k T^*M)$, and where we implicitly extend $\chi$ to $\Gamma(\wedge^k T^*M)$ by
    \begin{equation*}
    \chi(\eta) = \chi(\eta_1\wedge\ldots\wedge \eta_k):= (-1)^{\eta_1(k-1) + \eta_2(k-2)+\cdots + \eta_{k-1}}\chi_{\eta_1}\cdots\chi_{\eta_k}.
    \end{equation*}
We draw attention to the sign which arises when considering compositions of insertion operators under the map eq. \eqref{Section 1 Isomorphism}
    \begin{equation}\label{Section 3 Sign appearing from grouped insertion operators}
    \chi(\imath_{X_1}\ldots\imath_{X_n}\eta):= (-1)^{nX_1 +(n-1)X_2+\cdots + X_{n}}[\chi_{X_1},\ldots,[\chi_{X_n},\chi_\eta]\cdots].
    \end{equation}
We make similar definitions for sections of $\wedge TM$, except remark on the minus sign which appears from our choice of $\Pi TM$ as the base of the symplectic $2$-manifold,
    \begin{gather*}
    \chi(\imath_\eta X):=-(-1)^\eta[\chi_\eta,\chi_X],\qquad \chi(d_PX):= -[K_P,\chi_X],\\
    \chi(\C{L}_\eta X):=(-1)^\eta[[K_P,\chi_\eta],\chi_X].
    \end{gather*}

\begin{proposition}
Let $M$ be a homotopy Poisson manifold, and let $d_P$ denote the differential of multivector fields and $[\cdots]_{K_P}$ denote the higher Koszul brackets on differential forms. Then $\Gamma(\C{T})$ is naturally a homotopy Dorfman algebra, with brackets given by the formula
	\begin{align*}
	d_1(X) = -d^0_PX,\quad & \quad d_2(X,Y) = [X,Y]_c,\\
    d_2(X,\eta) = \C{L}_X\eta - (-1)^{X\eta}\imath_\eta d^1_PX,\quad & \quad  d_2(\eta,X) = \C{L}^*_\eta X - (-1)^{X\eta}\imath_Xd\eta,\\
    d_n(\eta_1,\ldots,\eta_n) = [\eta_1,&\ldots,\eta_n]_{K_P}, \mbox{ for all } n\geq 1,
    \end{align*}
and
    \begin{gather*}
    d_n(\eta_1,\ldots,\eta_{i-1},X,\eta_i,\ldots,\eta_n) = (-1)^{X(\eta_{i+1} + \cdots + \eta_n) + i}(2-i)\imath_{\eta_1}\cdots\imath_{\eta_n}d^n_PX\\
    - (-1)^{X(\eta_{i+1}+\cdots+\eta_n) + n+i}\sum^{i-1}_{j=1}(-1)^{\eta_j(\eta_{j+1}+\cdots+\eta_n) +j}\imath_{\eta_1}\cdots\widehat{\imath_{\eta_j}}\cdots\imath_{\eta_n}\C{L}^n_{\eta_j}X,
    \end{gather*}
for vector fields $X,Y$ and $1$-forms $\eta_1,\ldots,\eta_n$ on $M$, and where $d^n_PX$ is understood as the weight $n$ component of the differential $d_P = d_P^0 + d_P^1 + \cdots$, and similarly for $\C{L}_n$.
\end{proposition}

\begin{proof}
The function $\Delta$ caries the bigrading $2,1$, and a weight $k+1$ homogeneous component of $K_P$ has bigrading $1,k$. Considering the bigrading $1,0$ of a $1$-form $\chi_\eta$ and $0,1$ of a vector field $\chi_X$ in $\C{A}$, one can count the weights in eq. \eqref{Section 1 Definition Derived bracket formula} and deduce first that the brackets given are the only non-zero brackets. The brackets of two arguments are the Dorfman brackets for an arbitrary Lie bialgebroid and are well-known; the higher Koszul brackets appear from the structure of $T^*M$. The most interesting formula is the last. Consider the expression
    \begin{equation*}
    (-1)^{\epsilon}[\cdots [K_P,\chi_{\eta_1}],\dots,\chi_{\eta_{i-1}}],\chi_X],\chi_{\eta_{i+1}}],\ldots,\chi_{\eta_n}]
    \end{equation*}
where $\epsilon = \eta_1(n-1) + \eta_2(n-2) + \cdots + \eta_{i-1}(n-i+1) + X(n-i) + \eta_{i+1}(n-i-1) + \cdots + \eta_{n-1}$. To save space, introduce the notation $\underbrace{a}_{\delta} = (-1)^\delta a$. By applications of the Jacobi identity, the symmetry of the bracket, and observing when terms are zero, we obtain
    \begin{align*}
    (-1)^{\epsilon}& [\cdots [K_P,\chi_{\eta_1}],\dots,\chi_{\eta_{i-1}}],\chi_X],\chi_{\eta_{i+1}}],\ldots,\chi_{\eta_n}]\\
    & = \underbrace{[\cdots[[K_P,\chi_X],\chi_{\eta_1}],\ldots,\chi_{\eta_n}]}_{\epsilon + \chi_X(\chi_{\eta_1}+\cdots + \chi_{\eta_{i-1}})} + \sum^{i-1}_{j=1}\underbrace{[\cdots[[K,[\chi_{\eta_j},\chi_X]],\chi_{\eta_1}],\ldots,\chi_{\eta_n}]}_{\epsilon + \chi_X(\chi_{\eta_{i-1}} + \cdots + \chi_{\eta_{j+1}}) + (\chi_X+\chi_{\eta_j})(\chi_{\eta_{j-1}} + \cdots + \chi_{\eta_1})}\\
    & = \underbrace{(2-i)[\chi_{\eta_1},\ldots,[\chi_{\eta_n},[K_P,\chi_X]]]}_{\epsilon + n-1 + \chi_X(\chi_{\eta_1}+\cdots + \chi_{\eta_{i-1}}) + (\chi_X + 1)(\chi_{\eta_1} + \cdots + \chi_{\eta_n})}\\
    & \quad + \sum^{i-1}_{j=1}\underbrace{[\chi_{\eta_1},\ldots,\widehat{\chi_{\eta_j}},\ldots,[\chi_{\eta_n},[[K,\chi_{\eta_j}],\chi_X]]\cdots]}_{\epsilon + \chi_X(\chi_{\eta_{i-1}} + \cdots + \chi_{\eta_{j+1}}) + (\chi_X+\chi_{\eta_j})(\chi_{\eta_{j-1}} + \cdots + \chi_{\eta_1}) + n + (1+\chi_{\eta_j} + \chi_X)(\chi_{\eta_1}+\ldots +\hat{j} + \cdots + \chi_{\eta_n})}\\
    & = \underbrace{(2-i)[\chi_{\eta_1},\ldots,[\chi_{\eta_n},[K_P,\chi_X]]]}_{\epsilon + n-1 + \chi_X(\chi_{\eta_{i+1}}+\cdots + \chi_{\eta_n}) + \chi_{\eta_1} + \cdots + \chi_{\eta_n}}\\
    & \quad + \sum^{i-1}_{j=1}\underbrace{[\chi_{\eta_1},\ldots,\widehat{\chi_{\eta_j}},\ldots,[\chi_{\eta_n},[[K,\chi_{\eta_j}],\chi_X]]\cdots]}_{\epsilon + \chi_X(\chi_{\eta_{i+1}} + \cdots + \chi_{\eta_n}) + \chi_{\eta_j}(\chi_{\eta_{j+1}} + \cdots + \chi_{\eta_n}) + n + (\chi_{\eta_1}+\ldots +\hat{j} + \cdots + \chi_{\eta_n})}.
    \end{align*}
It remains to notice that to express this in terms of operators $\C{L}_\eta$ and $\imath_\eta$, we need the correct signs. Specifically, in the first term we need $n + \eta_1+\cdots+\eta_n$ for the $n-1$ insertion operators and the one differential $d_P$. Writing
    \begin{align*}
    & \epsilon + n-1 +  \chi_X(\chi_{\eta_{i+1}}+\cdots + \chi_{\eta_n}) + \chi_{\eta_1} + \cdots + \chi_{\eta_n}\\
    & = \epsilon + n-1 + (X + 1)(\eta_{i+1}+ + \cdots + \eta_n + n-i) + \eta_1 + \cdots + \eta_n + n-1\\
    & = (\epsilon + X(n-i) + \eta_{i+1}+ + \cdots + \eta_n) + (n + \eta_1 + \cdots + \eta_n) + X(\eta_{i+1}+ + \cdots + \eta_n) -i,
    \end{align*}
groups the terms into those altering $\epsilon$ to obtain the required sign to group insertion operators (see eq. \eqref{Section 3 Sign appearing from grouped insertion operators}), those required to write the bracket expression in terms of operators, and those remaining. Similar rearrangements for the second term give the required sign.
\end{proof}

\section{Proof of Main Statement}

\subsection{Polarisation}
For arbitrary elements of the vector space $\Gamma(E)$, showing that the higher Dorfman brackets of eq. \eqref{Section 1 Definition Derived bracket formula} satisfy the homotopy Loday identities is a lengthy procedure. One can significantly reduce the calculations by polarising the identities, which carry all the information along the diagonal. For an arbitrary section $u\in\Gamma(E)$, the identity $J_n(u) = J_n(u,\ldots,u)$ simplifies drastically (after considering the unshuffles). Conversely, the Loday identities on a collection of arbitrary sections can be recovered from $J_n(u)$ as follows.

For sections $u_1,\ldots,u_n\in\Gamma(E)$ and (super)numbers $\lambda^1,\ldots,\lambda^n$, writing $u = \lambda^iu_i$ allows the expression $J_n(u_1,\ldots,u_n)$ to be recovered as the coefficient of the monomial $\lambda^1\lambda^2\cdots\lambda^n$. In fact, the coefficient is not the single expression $J_n(u_1,\ldots,u_n)$, but rather it is proportional to the (skew-)symmetrisation of this, depending on the parity of $u$. Expanding $J_n(u) = J_n(\lambda^iu_i)$ in the monomial $\lambda^1\lambda^2\cdots\lambda^n$,
	\begin{equation*}
	\sum_{\sigma\in S_n}(-1)^{u(sgn(\sigma)) + \kappa(\sigma) + \varepsilon}\lambda^1\ldots\lambda^nJ_n(u_{\sigma(1)},\ldots,u_{\sigma(n)})
	\end{equation*}
where $\varepsilon$ is some sign depending on the parity of $u$, but not on the permutation $\sigma$. If $u$ is required to have even parity (i.e. $\tilde{\lambda}^i = \tilde{u}_i$) then this returns the symmetrisation of the expression $J_n(u_1,\ldots,u_n)$, and hence showing $J_n(u)=0$ for $\tilde{u}=0$ shows that either $J$ is zero, or that it is totally skew-symmetric. They are of course, not skew-symmetric operations due to the symmetry laws that the brackets possess, and hence it it enough to show that $J_n(u)=0$ for even $u$ to show that these identities hold on all sections. (Notice that for odd $u$, $J_n(u)$ are in fact symmetric operations, since $g$ is identically zero on the same odd arguments.)

\subsection{Proof}
We consider the case when $n=2m$. Using the polarisation trick described above, it will be sufficient to consider the $n$th higher Loday identity on a single even section $u\in\Gamma(E)$. In this case eq. \eqref{Section 1 nth Loday identity} reduces to
	\begin{align*}
	J_{2m}(u) = & \sum_{i+j=n+1}\sum_{k\geq j}\sum_{\sigma\in Sh(k-j,j-1)} sgn(\sigma)(-1)^{(k+1)(j+1)}\\
	& \qquad d_i(u_{\sigma(1)},\ldots,u_{\sigma(k-j)},d_j(u_{\sigma(k-j+1)},\ldots,u_{\sigma(k-1)},u_k),u_{k+1},\ldots,u_{i+j-1}),
	\end{align*}
where we then use the numbers $\C{C}(k,j)$ to denote the number of terms remaining after counting the signed unshuffles (see the appendix),
	\begin{equation*}
	= \sum_{i+j=2m+1}\sum^{2m}_{k= j}\C{C}(k,j)(-1)^{(k+1)(j+1)}d_i(\underbrace{u,\ldots,u}_{k-j},d_j(u,\ldots,u),\underbrace{u,\ldots,u}_{2m-k}).
	\end{equation*}
Now divide this summation into four; for $j=2p$, $j=2q+1$, write
	\begin{align*}
	& = \sum^m_{p=1}\sum^{2m}_{k= 2p}\C{C}(k,2p)(-1)^{k+1}d_i(\underbrace{u,\ldots,u}_{k-2p},d_{2p}(u,\ldots,u),\underbrace{u,\ldots,u}_{2m-k})\\
	&\quad + \sum^{m-1}_{q=0}\sum^{2m}_{k= 2q+1}\C{C}(k,2q+1)(-1)^{k}d_i(\underbrace{u,\ldots,u}_{k-2q-1},d_{2q+1}(u,\ldots,u),\underbrace{u,\ldots,u}_{2m-k}),
	\end{align*}
and again for $k$, write $k = 2r$ for the summation over $p$, and $k=2u+1$, $k=2v$ for that over $q$,
	\begin{align*}
	& = -\sum^m_{p=1}\sum^{m}_{r=p}\left(\begin{array}{c}r-1\\p-1\end{array}\right)d_i(\underbrace{u,\ldots,u}_{2(r-p)},d_{2p}(u,\ldots,u),\underbrace{u,\ldots,u}_{2(m-r)})\\
	&\qquad\quad + \sum^{m-1}_{q=0}\sum^{m-1}_{u=q}\left(\begin{array}{c}u\\q\end{array}\right)d_i(\underbrace{u,\ldots,u}_{2(u-q)},d_{2q+1}(u,\ldots,u),\underbrace{u,\ldots,u}_{2(m-u)-1})\\
	& \qquad\qquad \quad +\sum^{m-1}_{q=0}\sum^{m}_{v=q+1}\left(\begin{array}{c}v-1\\q\end{array}\right)d_i(\underbrace{u,\ldots,u}_{2(v-q)-1},d_{2q+1}(u,\ldots,u),\underbrace{u,\ldots,u}_{2(m-v)}).
	\end{align*}
The terms involving even $j$ and odd $k$ cancel from the sign of the shuffle, and so make no appearance in the expression. See the appendix for the source of the binomial coefficients appearing from the other unshuffles. By definition of the operations $d_k$, we translate the expression for the higher Loday identity to the Poisson algebra $\C{A}$ where we have at our disposal the Jacobi identity. In order to simplify the expressions, let us write $\xi = \chi_u$ where $\tilde{\xi}=1$, and
    \begin{equation*}
    \Phi^k_Q = [\cdots[Q\underbrace{\xi,\xi],\ldots,\xi}_{k  times}],\quad \Phi^0_Q =\theta,\quad\mbox{ and }\quad \chi_l = \chi_{d_l(u,\ldots,u)}.
    \end{equation*}
Separating out the terms $r=m$ and $v=m$, we obtain
    \begin{align*}
	& = -\sum^{m-1}_{p=1}\sum^{m-1}_{r=p}\left(\begin{array}{c}r-1\\p-1\end{array}\right)\left(P[\cdots[\Phi^{2(r-p)}_Q,\chi_{2p}],\underbrace{\xi],\ldots,\xi}_{2p}] - [P[\cdots[\Phi^{2(r-p)}_Q,\chi_{2p}],\underbrace{\xi],\ldots,\xi}_{2p}]\right)\\
	& \qquad -\sum^{m-1}_{p=1}\left(\begin{array}{c}m-1\\p-1\end{array}\right)\left(P[\Phi^{2(m-p)}_Q,\chi_{2p}] - [P\Phi^{2(m-p)}_Q,\chi_{2p}]\right)\\
	&- \sum^{m-1}_{q=0}\sum^{m-1}_{u=q}\left(\begin{array}{c}u\\q\end{array}\right)\left(P[\cdots[\Phi^{2(u-q)}_Q,\chi_{2q+1}],\underbrace{\xi],\ldots,\xi}_{2q+1}] - [P[\cdots[\Phi^{2(u-q)}_Q,\chi_{2q+1}],\underbrace{\xi],\ldots,\xi}_{2q+1}]\right)\\
	&+\sum^{m-2}_{q=0}\sum^{m-1}_{v=q+1}\left(\begin{array}{c}v-1\\q\end{array}\right)\left(P[\cdots[\Phi^{2(v-q)-1}_Q,\chi_{2q+1}],\underbrace{\xi],\ldots,\xi}_{2v}] - [P[\cdots[\Phi^{2(v-q)-1}_Q,\chi_{2q+1}],\underbrace{\xi],\ldots,\xi}_{2v}]\right)\\
	& \qquad + \sum^{m-1}_{q=0}\left(\begin{array}{c}m-1\\q\end{array}\right)\left(P[\Phi^{2(m-q)-1}_Q,\chi_{2q+1}] - [P\Phi^{2(m-q)-1}_Q,\chi_{2q+1}]\right),
	\end{align*}
where the sign $(-1)^{j(2m-k)}$ arises from eq. \eqref{Section 1 Definition Derived bracket formula} appearing only in the summation over $u$.

\begin{lemma}\label{Section Proof Lemma shuffling 2 xi}
For $\Phi^k_Q = [\cdots[Q\xi,\xi],\ldots,\xi]$ and $\tilde{\xi} =1$,
	\begin{equation*}
	[[[\Phi^k_Q,\chi_l],\xi],\xi] = [\Phi^{k+2}_Q,\chi_l],
	\end{equation*}
for any values of $k,l$.
\end{lemma}

\begin{proof}
By applications of the Jacobi identity and the skew-symmetry of the Poisson bracket,
	\begin{align*}
	& [[[\Phi^k_Q,\chi_l],\xi],\xi] = [[\Phi^k_Q,[\chi_l,\xi]],\xi] -(-1)^{\chi_l\Phi^k_Q}[[\chi_l,\Phi^{k+1}_Q],\xi]\\
&\quad  = (-1)^{\chi_l(\Phi^k_Q+\Phi^{k+1}_Q)}[\Phi^{k+1}_Q,[\chi_l,\xi]]-(-1)^{\Phi^k_Q(\chi_l+1)}[[\chi_l,\xi],\Phi^{k+1}_Q] - (-1)^{\Phi^k_Q\chi_l}[\chi_l,\Phi^{k+2}]\\
& \quad \quad = [\Phi^{k+2}_Q,\chi_l].
	\end{align*}
\end{proof}

Repeated applications of Lemma \ref{Section Proof Lemma shuffling 2 xi} and the binomial identities
    \begin{gather*}
	\sum^{m-1}_{r=p}\left(\begin{array}{c}r-1\\p-1\end{array}\right) = \sum^{m-1}_{r=p}\left(\left(\begin{array}{c}r\\p\end{array}\right) - \left(\begin{array}{c}r-1\\p\end{array}\right)\right) = \cdots = \left(\begin{array}{c}m-1\\p\end{array}\right),\\
	\sum^{m-1}_{u=q}\left(\begin{array}{c}u\\q\end{array}\right) = \left(\begin{array}{c}m\\q+1\end{array}\right),\quad \sum^{m-1}_{v=q+1}\left(\begin{array}{c}v-1\\q\end{array}\right) = \left(\begin{array}{c}m-1\\q+1\end{array}\right),
	\end{gather*}
allow the auxiliary indices $r,u,v$ to be eliminated from the summation. In doing so, we obtain the following expression for the $(2m)$th higher Loday identity in terms of the Poisson structure on $\C{A}$,
	\begin{align}\nonumber
	& = -\sum^{m-1}_{p=1}\left(\begin{array}{c}m-1\\p\end{array}\right)\left(P[\Phi^{2(m-p)}_Q,\chi_{2p}] - [P[[\Phi^{2(m-p-1)}_Q,\chi_{2p}],\xi],\xi]\right)\\\nonumber
	& \quad\qquad\qquad\qquad\qquad \qquad -\sum^{m}_{p=1}\left(\begin{array}{c}m-1\\p-1\end{array}\right)\left(P[\Phi^{2(m-p)}_Q,\chi_{2p}] - [P\Phi^{2(m-p)}_Q,\chi_{2p}]\right)\\\label{Section 2 Entire expression for Loday identity}
	&- \sum^{m-1}_{q=0}\left(\begin{array}{c}m\\q+1\end{array}\right)\left(P[[\Phi^{2(m-q-1)}_Q,\chi_{2q+1}],\xi] - [P[\Phi^{2(m-q-1)}_Q,\chi_{2q+1}],\xi]\right)\\\nonumber
	&\qquad \qquad+\sum^{m-2}_{q=0}\left(\begin{array}{c}m-1\\q+1\end{array}\right)\left(P[\Phi^{2(m-q)-1}_Q,\chi_{2q+1}] - [P[[\Phi^{2(m-q)-3}_Q,\chi_{2q+1}],\xi],\xi]\right)\\\nonumber
	& \quad\quad\qquad\qquad \qquad + \sum^{m-1}_{q=0}\left(\begin{array}{c}m-1\\q\end{array}\right)\left(P[\Phi^{2(m-q)-1}_Q,\chi_{2q+1}] - [P\Phi^{2(m-q)-1}_Q,\chi_{2q+1}]\right).
	\end{align}
It remains to write
	\begin{equation*}
	\chi_k = P\Phi^k_Q - [P\Phi^{k-1}_Q,\xi].
	\end{equation*}
Notice that the resulting sum will split into two: those terms whose projector lies on the outside of the expression, generating the weight $1$ terms, and those where the projector cannot be moved outside by manipulating the expressions, removing the weight $0$ terms. We split this into two and treat each part individually.

\subsubsection{The Outer Projector Terms}
 After substituting for $\chi_k$ in eq. \eqref{Section 2 Entire expression for Loday identity}, the terms containing the projector on the outside of the bracket, or those which can be manipulated to move it outside, are as follows
 	\begin{align}\nonumber
 	& -\sum^{m-1}_{p=1}\left(\begin{array}{c}m-1\\p\end{array}\right)\left(P[\Phi^{2(m-p)}_Q,P\Phi^{2p}_Q] \underbrace{- P[\Phi^{2(m-p)}_Q,[P\Phi^{2p-1}_Q,\xi]]}\right)\\\nonumber
 	& -\sum^{m}_{p=1}\left(\begin{array}{c}m-1\\p-1\end{array}\right)\left(P[\Phi^{2(m-p)}_Q,P\Phi^{2p}_Q] \underbrace{- P[\Phi^{2(m-p)}_Q,[P\Phi^{2p-1}_Q,\xi]]} - [P\Phi^{2(m-p)}_Q,P\Phi^{2p}_Q]\right)\\\label{Section 2 Outer projector terms large expression}
 	& - \sum^{m-1}_{q=0}\left(\begin{array}{c}m\\q+1\end{array}\right)\left(P[[\Phi^{2(m-q-1)}_Q,P\Phi^{2q+1}_Q],\xi] - P[[\Phi^{2(m-q-1)}_Q,[P\Phi^{2q}_Q,\xi]],\xi]\right)\\\nonumber
 	& + \sum^{m-2}_{q=0}\left(\begin{array}{c}m-1\\q+1\end{array}\right)\left(P[\Phi^{2(m-q)-1}_Q,P\Phi^{2q+1}_Q] - P[\Phi^{2(m-q)-1}_Q,[P\Phi^{2q}_Q,\xi]]\right)\\\nonumber
 	& + \sum^{m-1}_{q=0}\left(\begin{array}{c}m-1\\q\end{array}\right)\left(P[\Phi^{2(m-q)-1}_Q,P\Phi^{2q+1}_Q] - P[\Phi^{2(m-q)-1}_Q,[P\Phi^{2q}_Q,\xi]] - [P\Phi^{2(m-q)-1}_Q,P\Phi^{2q+1}_Q]\right).
 	\end{align}
By applying the Jacobi identity, the skew-symmetry of the bracket, and shifting index, we can rewrite
    \begin{gather}\nonumber
    - \sum^{m-1}_{q=0}\left(\begin{array}{c}m\\q+1\end{array}\right)\left(P[[\Phi^{2(m-q-1)}_Q,P\Phi^{2q+1}_Q],\xi] - P[[\Phi^{2(m-q-1)}_Q,[P\Phi^{2q}_Q,\xi]],\xi]\right)\\\label{Section 2 arranging m q+1 terms}
    \qquad = \sum^{m-1}_{q=0}\left(\begin{array}{c}m\\q+1\end{array}\right)\left(P[\Phi^{2(m-q)-1}_Q,[P\Phi^{2q}_Q,\xi]] - P[\Phi^{2(m-q)-1}_Q,P\Phi^{2q+1}_Q]\right)\\\nonumber
    - \sum^m_{p=1}\left(\begin{array}{c}m\\p\end{array}\right)P[\Phi^{2(m-p)}_Q,[P\Phi^{2p-1}_Q,\xi]].
    \end{gather}
Combining the two braced terms in eq. \eqref{Section 2 Outer projector terms large expression} via the identity
    \begin{equation*}
    \left(\begin{array}{c}m-1\\p\end{array}\right) + \left(\begin{array}{c}m-1\\p-1\end{array}\right) = \left(\begin{array}{c}m\\p\end{array}\right),
    \end{equation*}
we obtain
    \begin{equation}
    \sum^m_{p=1}\left(\begin{array}{c}m\\p\end{array}\right)P[\Phi^{2(m-p)}_Q,[P\Phi^{2p-1}_Q,\xi]],
    \end{equation}
which cancels with the second summation in eq. \eqref{Section 2 arranging m q+1 terms} leaving only the first.

Similarly we may combine the terms containing $\left(\begin{array}{c}m-1\\q+1\end{array}\right)$ in eq. \eqref{Section 2 Outer projector terms large expression} with those remaining in eq. \eqref{Section 2 arranging m q+1 terms} (after stripping out $q = m-1$ from the sum), to give
    \begin{gather}\nonumber
    \sum^{m-2}_{q=0}\left(\begin{array}{c}m-1\\q\end{array}\right)\left(P[\Phi^{2(m-q)-1}_Q,[P\Phi^{2q}_Q,\xi]] - P[\Phi^{2(m-q)-1}_Q,P\Phi^{2q+1}_Q]\right)\\\nonumber + P[\Phi^1,[P\Phi^{2(m-1)}_Q,\xi]] - P[\Phi^1_Q,P\Phi^{2m-1}_Q]\\\label{Section 2 the m-1 q terms which will cancel}
    = \sum^{m-1}_{q=0}\left(\begin{array}{c}m-1\\q\end{array}\right)\left(P[\Phi^{2(m-q)-1}_Q,[P\Phi^{2q}_Q,\xi]] - P[\Phi^{2(m-q)-1}_Q,P\Phi^{2q+1}_Q]\right).
    \end{gather}
Observe that the terms in eq. \eqref{Section 2 the m-1 q terms which will cancel} cancel precisely with those first and second terms in eq.  \eqref{Section 2 Outer projector terms large expression} containing $\left(\begin{array}{c}m-1\\q\end{array}\right)$, and hence after this cancellation eq. \eqref{Section 2 Outer projector terms large expression} reduces to
    \begin{gather}\label{Section 2 Reduced expression for the outer terms}
    -\sum^{m-1}_{p=1}\left(\begin{array}{c}m-1\\p\end{array}\right)P[\Phi^{2(m-p)}_Q,P\Phi^{2p}_Q] - \sum^{m}_{p=1}\left(\begin{array}{c}m-1\\p-1\end{array}\right)P[\Phi^{2(m-p)}_Q,P\Phi^{2p}_Q]\\\nonumber + \sum^{m}_{p=1}\left(\begin{array}{c}m-1\\p-1\end{array}\right)[P\Phi^{2(m-p)}_Q,P\Phi^{2p}_Q] - \sum^{m-1}_{q=0}\left(\begin{array}{c}m-1\\q\end{array}\right)[P\Phi^{2(m-q)-1}_Q,P\Phi^{2q+1}_Q].
    \end{gather}
We now use the distributivity of the projector $P$ (eq. \eqref{Section 1 Distr Law for projection}), and write
    \begin{gather*}
    - \sum^{m}_{p=1}\left(\begin{array}{c}m-1\\p-1\end{array}\right)P[\Phi^{2(m-p)}_Q,P\Phi^{2p}_Q] + \sum^{m}_{p=1}\left(\begin{array}{c}m-1\\p-1\end{array}\right)[P\Phi^{2(m-p)}_Q,P\Phi^{2p}_Q]\\
    = \sum^{m}_{p=1}\left(\begin{array}{c}m-1\\p-1\end{array}\right)P[P\Phi^{2(m-p)}_Q,\Phi^{2p}_Q] - \sum^{m}_{p=1}\left(\begin{array}{c}m-1\\p-1\end{array}\right)P[\Phi^{2(m-p)}_Q,\Phi^{2p}_Q],
    \end{gather*}
for the summation over $p$, together with
    \begin{gather*}
    - \sum^{m-1}_{q=0}\left(\begin{array}{c}m-1\\q\end{array}\right)[P\Phi^{2(m-q)-1}_Q,P\Phi^{2q+1}_Q] = \sum^{m-1}_{q=0}\left(\begin{array}{c}m-1\\q\end{array}\right)P[\Phi^{2(m-q)-1}_Q,\Phi^{2q+1}_Q]\\
    -\sum^{m-1}_{q=0}\left(\begin{array}{c}m-1\\q\end{array}\right)P[P\Phi^{2(m-q)-1}_Q,\Phi^{2q+1}_Q] - \sum^{m-1}_{q=0}\left(\begin{array}{c}m-1\\q\end{array}\right)P[\Phi^{2(m-q)-1}_Q,P\Phi^{2q+1}_Q],
    \end{gather*}
for that over $q$.  A shift in index shows that the final two terms of this expression cancel identically, and eq. \eqref{Section 2 Reduced expression for the outer terms} becomes
    \begin{gather*}
    -\sum^{m-1}_{p=1}\left(\begin{array}{c}m-1\\p\end{array}\right)P[\Phi^{2(m-p)}_Q,P\Phi^{2p}_Q] - \sum^{m}_{p=1}\left(\begin{array}{c}m-1\\p-1\end{array}\right)P[\Phi^{2(m-p)}_Q,\Phi^{2p}_Q]\\ + \underbrace{\sum^{m}_{p=1}\left(\begin{array}{c}m-1\\p-1\end{array}\right)P[P\Phi^{2(m-p)}_Q,\Phi^{2p}_Q]} + \sum^{m-1}_{q=0}\left(\begin{array}{c}m-1\\q\end{array}\right)P[\Phi^{2(m-q)-1}_Q,\Phi^{2q+1}_Q].
    \end{gather*}
If one sets $s = m-p$ in the braced term, the symmetry of the bracket negates the first term directly above it (notice that the term $s=0$ is removed by the projection condition $P\Phi^0 = 0$), and we are left with
    \begin{equation}\label{Section 2 Very reduced outer term expression}
    = \sum^{m-1}_{q=0}\left(\begin{array}{c}m-1\\q\end{array}\right)P[\Phi^{2(m-q)-1}_Q,\Phi^{2q+1}_Q] - \sum^{m}_{p=1}\left(\begin{array}{c}m-1\\p-1\end{array}\right)P[\Phi^{2(m-p)}_Q,\Phi^{2p}_Q].
    \end{equation}
Changing index $p = r+1$ in the second term and writing $\Phi^{2(r+1)}_Q = [\Phi^{2r+1}_Q,\xi]$ allows the use of the Jacobi identity, which provides the required term to remove the first in eq. \eqref{Section 2 Very reduced outer term expression}. The term remaining is
    \begin{equation}\label{Section 2 End equation for outer terms}
    = -\sum^{m-1}_{r=0}\left(\begin{array}{c}m-1\\r\end{array}\right)P[[\Phi^{2(m-r-1)}_Q,\Phi^{2r+1}_Q],\xi].
    \end{equation}

Finally, it remains to move $Q$ from one term in $\Phi$ and push it into the second, obtaining $P\Phi^{2m}_{Q^2}$ which generates the derived bracket of $Q^2$. Write eq. \eqref{Section 2 End equation for outer terms} as
	\begin{equation*}
	- P[Q\Phi^{2m-1}_Q,\xi] - \sum^{m-2}_{q=0}\left(\begin{array}{c}m-1\\q\end{array}\right)P[[\Phi^{2(m-q-1)}_Q,\Phi^{2q+1}_Q],\xi],
	\end{equation*}
where $Q$ appears as $\Phi^0$, and consider only the first term $P[Q\Phi^{2m-1}_Q,\xi]$. Applications of the Jacobi identity give
	\begin{gather}\nonumber
	P[Q\Phi^{2m-1}_Q,\xi] = P\Phi^{2m}_{Q^2} - \sum^{2m-1}_{k=2}(-1)^k[\cdots[\Phi^{2m-k}_Q,\Phi^1_Q],\underbrace{\xi],\ldots,\xi}_{k-1}]\\\label{Section 2 Identity for the q squared bracket cancellation}
	= P\Phi^{2m}_{Q^2}  -\sum^{m-1}_{p=1}[\cdots[\Phi^{2(m-p)}_Q,\Phi^1_Q],\underbrace{\xi],\ldots,\xi}_{2p-1}] + \sum^{m-1}_{q=1}[\cdots[\Phi^{2(m-q)-1}_Q,\Phi^1_Q],\underbrace{\xi],\ldots,\xi}_{2q}]
	\end{gather}
upon splitting the sum into even and odd $k$.

\begin{lemma}\label{Section 2 lemma for q squared}
The term $\Phi^k_Q$ satisfies
	\begin{equation*}
	[[[\Phi^k_Q,\Phi^l_Q],\xi],\xi] = [\Phi^k_Q,\Phi^{l+2}_Q] + [\Phi^{k+2}_Q,\Phi^{l}_Q],
	\end{equation*}
and more generally,
	\begin{equation*}
	[\cdots[\Phi^k_Q,\Phi^l_Q],\underbrace{\xi],\ldots,\xi}_{2r}] = \sum^{r}_{s=0}\left(\begin{array}{c}r\\s\end{array}\right)[\Phi^{k+2s}_Q,\Phi^{l+2(r-s)}_Q].
	\end{equation*}
\end{lemma}
The proof is an exercise in the Jacobi identity and we omit it. Applying the result of Lemma \ref{Section 2 lemma for q squared} to eq. \eqref{Section 2 Identity for the q squared bracket cancellation}, gives
	\begin{align}\nonumber
	 = P\Phi^{2m}_{Q^2} - & \sum^{m-1}_{p=1}\sum^{p-1}_{r=0}\left(\begin{array}{c}p-1\\r\end{array}\right)P[[\Phi^{2(m-p)+2r}_Q,\Phi^{1+2(p-r-1)}_Q],\xi]\\\nonumber
	& \qquad  + \sum^{m-1}_{q=1}\sum^{q-1}_{s=0}\left(\begin{array}{c}q-1\\s\end{array}\right)P[[[\Phi^{2(m-q)+2s-1}_Q,\Phi^{1+2(q-s-1)}_Q],\xi],\xi]\\\label{Section 2 Almost final reference}
	 = P\Phi^{2m}_{Q^2}  + & \sum^{m-1}_{q=1}\sum^{q-1}_{s=0}\left(\begin{array}{c}q-1\\s\end{array}\right)P[[\Phi^{2(m-q+s)-1}_Q,\Phi^{2(q-s)}_Q],\xi],
	\end{align}
after an application of the Jacobi identity to the summation over $s$, which cancels that in $r$.

Considering the sum eq. \eqref{Section 2 End equation for outer terms} with the terms in eq. \eqref{Section 2 Almost final reference}, the remaining terms are
	\begin{align*}
	- P\Phi^{2m}_{Q^2} - \sum^{m-1}_{q=1}\sum^{q-1}_{s=0}\left(\begin{array}{c}q-1\\s\end{array}\right)&P[[\Phi^{2(m-q+s)-1}_Q,\Phi^{2(q-s)}_Q],\xi]\\ & - \sum^{m-2}_{q=0}\left(\begin{array}{c}m-1\\q\end{array}\right)P[[\Phi^{2(m-q-1)}_Q,\Phi^{2q+1}_Q],\xi].
	\end{align*}
Setting $t = q-s$ in the first summation,
	\begin{align*}
	\sum^{m-1}_{q=1}\sum^{q-1}_{s=0}&\left(\begin{array}{c}q-1\\s\end{array}\right)P[[\Phi^{2(m-q+s)-1}_Q,\Phi^{2(q-s)}_Q],\xi]\\
	& = -\sum^{m-1}_{q=1}\sum^{q}_{t=1}\left(\begin{array}{c}q-1\\t\end{array}\right)P[[\Phi^{2t}_Q,\Phi^{2(m-t)-1}_Q],\xi]\\
	&\qquad\qquad  = -\sum^{m-1}_{t=1}\sum^{m-1}_{q=t}\left(\begin{array}{c}q-1\\t-1\end{array}\right)P[[\Phi^{2t}_Q,\Phi^{2(m-t)-1}_Q],\xi]\\
	& \qquad\qquad\qquad \qquad = -\sum^{m-1}_{t=1}\left(\begin{array}{c}m-1\\t\end{array}\right)P[[\Phi^{2t}_Q,\Phi^{2(m-t)-1}_Q],\xi],
	\end{align*}
for which a final change of index shows that the only term remaining is the term involving $Q^2$,
	\begin{equation*}
	= - P\Phi^{2m}_{Q^2}.
	\end{equation*}

\subsubsection{The Inner Projector terms} In a similar fashion to the terms containing the projector $P$ on the outside of the expression, we collect and rearrange those terms whose projector lies on the inside to obtain the result. These are
	\begin{align*}
	& = \sum^{m-1}_{p=1}\left(\left(\begin{array}{c}m-1\\p\end{array}\right)[P[[\Phi^{2(m-p-1)}_Q,P\Phi^{2p}_Q],\xi],\xi] - [P[[\Phi^{2(m-p-1)}_Q,[P\Phi^{2p-1}_Q,\xi]],\xi],\xi]\right)\\
	& \qquad- \sum^{m-1}_{p=1}\left(\begin{array}{c}m-1\\p-1\end{array}\right)[P\Phi^{2(m-p)}_Q,[P\Phi^{2p-1}_Q,\xi]]\\
	& + \sum^{m-1}_{q=0}\left(\left(\begin{array}{c}m\\q+1\end{array}\right)[P[\Phi^{2(m-q-1)}_Q,P\Phi^{2q+1}_Q],\xi] - [P[\Phi^{2(m-q-1)}_Q,[P\Phi^{2q}_Q,\xi]],\xi]\right)\\
	& - \sum^{m-2}_{q=0}\left(\left(\begin{array}{c}m-1\\q+1\end{array}\right)[P[[\Phi^{2(m-q)-3}_Q,P\Phi^{2q+1}_Q],\xi],\xi] - [P[[\Phi^{2(m-q)-3}_Q,[P\Phi^{2q}_Q,\xi]],\xi],\xi]\right)\\
	& \qquad+\sum^{m-1}_{q=1}\left(\begin{array}{c}m-1\\q\end{array}\right)[P\Phi^{2(m-q)-1}_Q,[P\Phi^{2q+1}_Q,\xi]].
	\end{align*}
The manipulation of these terms follows closely those with the projector on the outside, and to reduce what are already long calculations we will not repeat these. We claim that by repeating similar steps,  we can reduce this long expression to
	\begin{equation*}
	[PQ\Phi^{2m-1}_Q,\xi] + \sum^{m-1}_{q=0}\left(\begin{array}{c}m-1\\q\end{array}\right)[P[\Phi^{2(m-q-1)}_Q,\Phi^{2q+1}_Q],\xi].
	\end{equation*}
This expression is almost identical to that obtained from eq. \eqref{Section 2 End equation for outer terms}, and following that which follows eq. \eqref{Section 2 End equation for outer terms}, this reduces to
	\begin{equation*}
	[P\Phi^{2m-1}_{Q^2},\xi].
	\end{equation*}
Thus the entire higher Loday identity is equal to
	\begin{equation*}
	\chi(J_n(u)) = - P\Phi^{2m}_{Q^2} + [P\Phi^{2m-1}_{Q^2},\xi],
	\end{equation*}
which is precisely (the negative of) the $n$th bracket for the derivation $Q^2$.

These calculations show that the result holds for even $n$. The case for the odd $n$ is analogous, and involves altering the summation indices together with some of the signs. It is not very beneficial to repeat the case of odd $n$, and we leave it as an exercise for the interested reader.

\section*{Appendix: Counting Signed Unshuffles}
After the use of the polarisation trick for even elements in the higher Loday identities eq. \eqref{Section 1 nth Loday identity}, we are left with calculating the sum of the signed unshuffles of identical elements. Due to the presence of the sign, many of these cancel, and we are left with a combinatorial problem of how to calculate the number of remaining terms. It proves rather difficult to track down such a problem in the literature, and so we use the Appendix to provide a self contained proof.

Let $\{x_1,\ldots,x_k\}$ be a set of elements labeled $1$ to $k$, and let $1\leq j\leq k$ be a fixed integer. For our purposes, we wish to fix $x_k$, and consider the set of $(k-1,j-1)$-unshuffles: permutations $\sigma\in S_{k-1}$ such that
	\begin{equation*}
	x_{\sigma(1)}<\cdots < x_{\sigma(k-j)},\qquad x_{\sigma(k-j+1)}<\cdots < x_{\sigma(k-1)},
	\end{equation*}
preserving the numerical order of the labels of the first $k-j$ elements and the final $j-1$ elements. We will write the permutation as
	\begin{equation*}
	x_1,\ldots, x_k \mapsto (-1)^{sgn(\sigma)}x_{\sigma(1)},\ldots, x_{\sigma(k-j)},(x_{\sigma(k-j+1)},\ldots,x_{\sigma(k-1)},x_k),
	\end{equation*}
where $sgn(\sigma)$ is $0$ or $1$ depending on the parity of the permutation (even if it is composed of an even number of transpositions and so on). 

Write the binomial coefficient as
	\begin{equation*}
	C^n_m := \left(\begin{array}{c}n\\m\end{array}\right).
	\end{equation*}
It is easily deduced that the total number of $(k-1,j-1)$ unshuffles is given by $C^{k-1}_{j-1}$. Our question is how many of these unshuffles are even (as permutations), how many are odd, and what is the difference. Let
	\begin{equation*}
	\C{C}(k,j) = \# E C^{k-1}_{j-1} - \# O C^{k-1}_{j-1}
	\end{equation*}
denote the difference between the number of even unshuffles and the number of odd unshuffles. A priori, all we can say is that $0\leq |\C{C}(k,j)|\leq C^{k-1}_{j-1}$, but in fact $0\leq \C{C}(k,j)\leq C^{k-1}_{j-1}$ follows after a little investigation.

\begin{theorem}
The number $\C{C}(k,j)$ of even unshuffles minus those which are odd is given by the following:
	\begin{align*}
	C^{p-1}_{q-1} & \qquad   k=2p, j = 2q, \\ 0\quad & \qquad k=2p+1, j=2q,\\ C^{p-1}_q & \qquad k=2p, j=2q+1,\\ C^p_q\,\,\, & \qquad k=2p+1, j=2q+1.
	\end{align*}
\end{theorem}
We see that these numbers depend strongly on the parity carried by $k$ and $j$, and in only one case do the number of even and odd permutations coincide. These can be reasonably checked for $k$ up to $9$ or $10$, simply by writing down the possible shuffles and determining their parity. This is useful, since they are required for the inductive procedure needed to prove the result, whose inductive step depends on the knowledge of all previous numbers.

\begin{proof}
Let us fix an arbitrary $k$ and $j$. The total number of unshuffles $C^{k-1}_{j-1}$ can be decomposed by fixing certain elements, and offers a combinatorial interpretation of the so called Hockey Stick identity for binomial coefficients. To begin, the total number of unshuffles is equal to the sum of the number of unshuffles fixing $x_1$ as the first element,  with the number of those which do not. Let us fix $x_1$ as the first term:
	\begin{equation*}
	x_1,x_{\sigma(2)},\ldots, x_{\sigma(k-j)},(x_{\sigma(k-j+1)},\ldots,x_{\sigma(k-1)},x_k).
	\end{equation*}
There are $k-2$ free elements, and $k-j-1$ available pairings, hence there are $C^{k-2}_{k-j-1}$ terms with $x_1$ arising as the first term. Consider then those beginning with the element $x_2$.	 Notice that if $x_2$ is fixed as the first term, $x_1$ must appear in the $j$-partition, specifically as the first term,
	\begin{equation*}
	x_2,x_{\sigma(3)},\ldots, x_{\sigma(k-j+1)},(x_1,x_{\sigma(k-j+2)},\ldots,x_{\sigma(k-1)},x_k).
	\end{equation*}
This leaves $k-3$ free elements, and still $k-j-1$ available options. Hence there are $C^{k-3}_{k-j-1}$ unshuffles containing $x_2$ as the first term. Now fix $x_r$ as the first term with $1\leq r\leq j$. This forces $x_1,\ldots,x_{r-1}$ into the $j$-partition, preserving the order, and thus ensures that no term $x_i$ with $i\geq j+1$ appears as the first term in the unshuffle. Similar reasoning then shows that there are $C^{k-1-r}_{k-j-1}$ unshuffles with $x_r$ fixed as the first term, and hence
	\begin{equation*}
	C^{k-1}_{j-1} = C^{k-2}_{k-j-1} + C^{k-3}_{k-j-1} + \cdots + C^{k-j-1}_{k-j-1} = \sum^j_{r=1}C^{k-1-r}_{k-j-1},
	\end{equation*}
recovering the well-known identity.
	
Let us now turn to those shuffles which are even or odd. Given a $(k-1,j-1)$ unshuffle $\sigma$,
	\begin{align*}
	x_1,& \ldots, x_k\\
	& = (-1)^{sgn(\sigma)}x_r,x_{\sigma(r+1)},\ldots,x_{\sigma(k-j+r-1)},x_1,\ldots,x_{r-1},x_{\sigma(k-j+r)},\ldots,x_{\sigma(k-1)},x_k\\
	& = (-1)^{(r-1)(k-j)+sgn(\sigma)}x_1,\ldots, x_r, x_{\sigma(r+1)},\ldots,x_{\sigma(k-1)},x_k,
	\end{align*}
where the second equality is obtained from the first by moving the $r-1$ elements $x_1,\ldots,x_{r-1}$ past $k-j$ elements, using $(r-1)(k-j)$ transpositions. The parity of the unshuffle $\sigma$ fixing $x_r$ as the first element then can be obtained directly from the parity of the unshuffle  which permutes only $k-r-1$ elements. Specifically, depending on $k,j,r$, the larger even unshuffles can be recovered from the smaller even and odd. We now work case by case.

Fix $k=2p$, $j=2q$, and suppose that the theorem holds for all $k$ and $j$ less than these values. Then
	\begin{equation*}
	\C{C}(2p,2q) = \# E C^{2p-1}_{2q-1} - \# O C^{2p-1}_{2q-1} = \# E \sum^{2q}_{r=1}C^{2p-1-r}_{2(p-q)-1} - \# O \sum^{2q}_{r=1}C^{2p-1-r}_{2(p-q)-1}.
	\end{equation*}
Since $k-j = 2(p-q)$ is even, the number of even/ odd permutations of the sum is equal to the sum of the number of even/ odd permutations of each individual part. Hence
	\begin{align*}
	\C{C}(2p,2q) & = \sum^{2q}_{r=1}\C{C}(2p-r,2(p-q))\\
	& \quad = \sum^q_{u=1} \C{C}(2(p-u),2(p-q)) + \sum^{q-1}_{v=0} \C{C}(2(p-v-1)+1,2(p-q))\\
	&  \quad \quad = \sum^q_{u=1} C^{p-u-1}_{p-q-1} = C^{p-1}_{p-q} = C^{p-1}_{q-1},
	\end{align*}
where we use the first results of the theorem as assumed for induction for even $j$.
	
The proof for the case $k=2p+1$, $j=2q+1$ is similar, since $k-j$ remains even. We find that 
	\begin{align*}
	\C{C}(2p+1,2q+1) & = \sum^q_{u=1} \C{C}(2(p-u)+1,2(p-q)) + \sum^{q}_{v=0} \C{C}(2(p-v)+1,2(p-q))\\
	& \quad = \sum^{q}_{v=0} \C{C}(2(p-v)+1,2(p-q)) = C^{p}_{p-q} = C^p_q,
	\end{align*}	
where we assume the vanishing of the terms in the summation over $u$.
	
The difference comes in the terms where $k-j$ is odd, and so the parity of the shuffles obtained from the lower terms changes according to the value of $r$. Fix $k= 2p+1$, $j=2q$. Then
	\begin{equation*}
	\# E C^{2p}_{2q-1} = \# E \sum^{2q}_{r-1}C^{2p-r}_{2(p-q)}.
	\end{equation*}
If $r$ is odd then $(r-1)(k-j)$ is even, and so the even shuffles of $C^{2p}_{2q-1}$ are given by the even shuffles of the terms in the sum in odd $r$. If $r$ is even however, the parity of the corresponding permutation shifts, and those even shuffles are given by the odd ones. We obtain
	\begin{align*}
	\C{C}(2p+1,2q) &  =\sum^{2q}_{r=1}(-1)^{r+1}\C{C}(2p+1-r,2(p-q)+1)
	\end{align*}
Splitting around even and odd $r$, and using the hypothesis for those lower terms in odd $j$, we see that this vanishes identically. The final case for $k=2p$, $j=2q+1$ follows similar reasoning, and we recover
	\begin{equation*}
	\C{C}(2p,2q+1) = C^{p-1}_{q}.
	\end{equation*}
\end{proof}

\subsection*{Acknowledgments}
I am very grateful to Ted Voronov who first suggested that such homotopy structure should exist incorporating the Dorfman bracket, and for the many discussions, both related and unrelated, over the past years.

\bibliographystyle{alpha}
\bibliography{bibliographybypickles}

\end{document}